\documentclass[a4paper,english,numberwithinsect]{eurocg24}
\usepackage{xspace}
\usepackage{booktabs}
\usepackage{microtype} 
\usepackage{mathtools}

\newcommand{\cR}{\mathcal{R}}

\newcommand{\ed}{($\epsilon,\delta$)}
\newcommand{\VCd}{VC dimension\xspace}

\theoremstyle{proposition}
\newtheorem{proposition}[theorem]{Proposition}

\title{Approximating Simplet Frequency Distribution for Simplicial Complexes\footnote{Work by the second and fourth authors is partially supported by the European Research Council (ERC), grant no.\ 788183, by the Wittgenstein Prize, Austrian Science Fund (FWF), grant no.\ Z 342-N31, and by the DFG Collaborative Research Center TRR 109, Austrian Science Fund (FWF), grant no.\ I 02979-N35.}}

\titlerunning{Approximating Simplet Frequency Distribution for Simplicial Complexes}

\author[1]{Hamid Beigy}
\author[2]{Mohammad Mahini}
\author[3]{Salman Qadami}
\author[4]{Morteza Saghafian}

\affil[1]{Sharif University of Technology\\
  \texttt{beigy@sharif.edu}}
  
\affil[2]{Sharif University of Technology\\        
\texttt{m\_mahini@ce.sharif.edu}}

\affil[3]{Amirkabir University of Technology\\        
\texttt{salmanqadami@gmail.com}}

\affil[4]{Institute of Science and Technology Austria\\
    \texttt{morteza.saghafian@ist.ac.at}}    
\authorrunning{H. Beigy, M. Mahini, S. Qadami, and M. Saghafian}

\ArticleNo{06}

\begin{document}

\maketitle

\begin{abstract}
Simplets, constituting elementary units within simplicial complexes (SCs), serve as foundational elements for the structural analysis of SCs.
Previous efforts have focused on the exact count or approximation of simplet count rather than their frequencies, with the latter being more practical in large-scale SCs.
This paper enables simplet frequency analysis of SCs by introducing the Simplet Frequency Distribution (SFD) vector. In addition, we present a bound on the sample complexity required for accurately approximating the SFD vector by any uniform sampling-based algorithm.
We also present a simple algorithm for this purpose and justify the theoretical bounds with experiments on some random simplicial complexes.
\end{abstract}

\section{Introduction}
In a range of disciplines, including biology, geology, and social science, the application of simplicial complexes is frequently employed to extract essential structural insights.
\emph{Simplicial Complexes (SCs)} are defined as networks of higher-order that possess the property of downward closure, which makes them suitable for representing higher-order relationships within network-like structures and their geometrical aspects~\cite{bianconi2021higher,edelsbrunner2014short,jonsson2008simplicial}.
In particular, SCs are used to study the geometric and combinatorial structure of protein interaction networks~\cite{estrada2018centralities}, epidemic spreading~\cite{li2021contagion}, co-authorship relations~\cite{sinha2015overview}, analyze email communications~\cite{leskovec2007graph}, and investigate the functional and structural organization of the brain~\cite{lord2016insights}.

Analyzing network behavior using small network building blocks, commonly known as \emph{motifs}, is common in numerous fields, including biological~\cite{alon2007network} and social networks~\cite{rotabi2017detecting}. Graphs are great examples where researchers use small building blocks called \emph{graphlets} to understand how networks behave based on local structures~\cite{ribeiro2021survey}. Graphlet analysis has many applications in biological networks~\cite{douglas2022exploring,windels2022graphlet}, and social networks~\cite{ashford2022understanding,baas2018predicting}.
By considering \emph{simplets} as fundamental elements within simplicial complexes, analogous to graphlets in the context of SCs, we can examine the specific patterns formed by the simplices associated with different sets of nodes~\cite{preti2022fresco}. This approach offers a straightforward way of analyzing complex networks' structural characteristics and their constituent parts.



\subparagraph*{Approximating Graphlet Count and Distribution.}

Numerous investigations have delved into the precise enumeration of graphlet types or approximating their frequencies. Several studies, like the ESU and RAGE algorithms, count the precise number of graphlets~\cite{wernicke2006efficient,marcus2012rage}.
Meanwhile, various algorithms like GRAFT, CC have employed sampling techniques to estimate the frequency of graphlets~\cite{bhuiyan2012guise,bressan2023efficient,bressan2017counting,bressan2018motif}.
For instance, Bressan in \cite{bressan2023efficient} introduced a random walk based method that preprocesses $k$-vertex graphlets, and gives a random graphlet in the time complexity of $k^{O(k)} \cdot \log{\Delta}$, where $\Delta$ is the maximum degree in the given graph.

\subparagraph*{Approximating Simplet Count and Distribution.}
Preti et al. introduced the concept of simplets, and the FRESCO algorithm that indirectly estimates the quantity of each simplet by utilizing a proxy metric referred to as \emph{support}~\cite{preti2021strud, preti2022fresco}.
B-Exact precisely enumerates up to 4-node configurations through combinatorial techniques~\cite{benson2018simplicial}. Importantly, each simplet can correspond to zero, one, or more than one configuration.
Kim et al. presented SC3, a sampling-based algorithm for approximating simplet count, that utilizes color coding techniques~\cite{kim2023characterization}.
Thus far, a limited number of dedicated algorithms designed for counting simplets either precisely or approximately. The concept of simplets is relatively novel, and numerous opportunities remain untapped for their application in various contexts.
\subparagraph*{Contribution.}
Calculating the exact quantity of each graphlet type or simplet type is frequently prohibitively expensive, and for numerous practical purposes, obtaining an estimated count of various graphlet types and simplet types or approximating their frequency distribution is usually adequate.
This paper studies the concept of the \emph{Simplet Frequency Distribution (SFD)} for the first time (to the best of the authors' knowledge), which can be more practical in analyze of large-scale SCs. Alongside this new concept, we present an algorithm  to approximating the SFD vector based on uniform sampling of simplets. 

More importantly, we present an upper-bound on the sample complexity (number of samples needed) of any approximation algorithm based on a uniform sampling method.
By doing this, we aim to enhance our comprehension and analysis of simplicial complexes, mapping them to vector spaces and using this vector for machine learning applications such as classification.
In overview, we present the following contributions.
\begin{itemize}
    \item Defining the concept of the Simplet Frequency Distribution (SFD) vector
    \item Studying an upper bound on the number of samples we need for every sampling based algorithm for approximating the SFD vector
    \item Proposing an algorithm for approximating the SFD vector by uniform sampling of simplets
\end{itemize}

\section{Preliminaries}
\label{sec:Concepts-Definitions} 
Within this section, we lay out the foundational concepts employed in this paper.
\begin{figure*}
    \begin{center}
        \centerline{\includegraphics[width=1\textwidth,keepaspectratio]{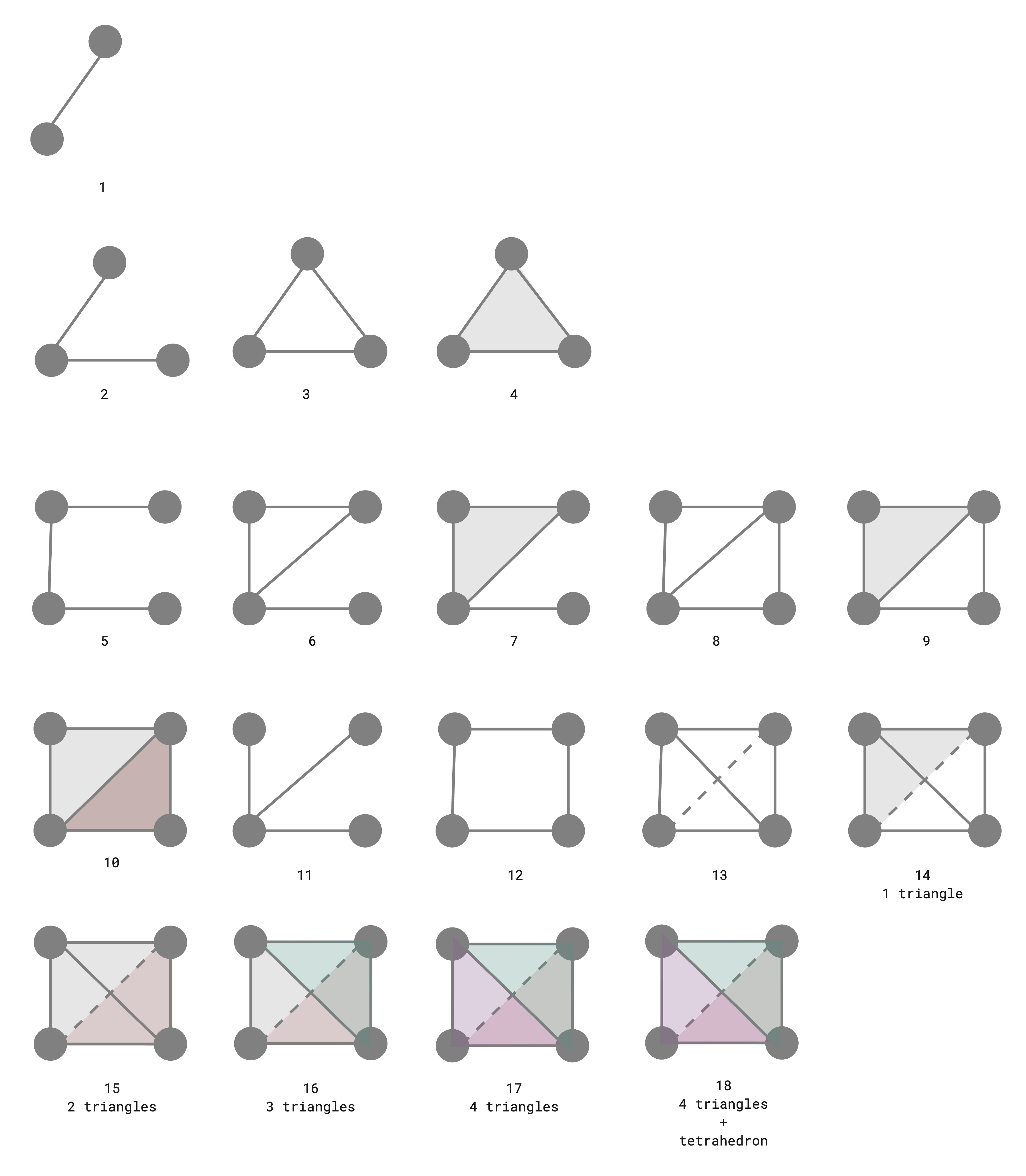}}
        
	\caption{The set of all $18$ simplet types with at least two and at most four vertices.} 
	\label{fig:simplets}
    \end{center}
\end{figure*}
\subparagraph*{Simplex.}
A $n$-simplex is the convex hull of $n + 1$ distinct points in $n$-dimensional space. A face of an $n$-simplex $\sigma$ is the convex hull of any non-empty subset $S$ of its vertices.
\subparagraph*{Simplicial Complex.}
A simplicial complex $\mathcal {K}$ is a set of simplices that is closed under taking faces, and the non-empty intersection of any two simplices $\sigma,\tau\in \mathcal{K}$ is a face of both $\sigma,\tau$. Simplicial complexes provide a combinatorial and topological framework for studying the structure of spaces through simplices, capturing both geometric and connectivity information.

\subparagraph*{Simplet.}
\emph{Simplets} are small induced connected sub-complexes of a massive complex that appear at any frequency.
A complex $H$ is an induced sub-complex of $\mathcal{K}$ if and only if, for any simplex $S$ in $\mathcal{K}$ whose vertices are a subset of $V(H)$, $S$ should also be in $H$.
So, every simplet can be identified by its vertices, typically regarded as being at least two. A \emph{simplet set} is a set of simplets of a simplicial complex. 
\emph{Simplet types} are isomorphic classes of simplets.
We denote $\mathcal{S}_{\mathcal{K}}(i)$ as a set of all simplets of type $i$ in $\mathcal{K}$, where $1 \leq i \leq N_m$, and $N_m$ is the number of simplet types with at most $m$ vertices.
Also, we denote $\mathcal{S}_{\mathcal{K}}^m$ as the set of all simplets in $\mathcal{K}$ with at most $m$ vertices. We assume that $m$ is a constant small number.

\subparagraph*{Simplet Frequency Distribution.}
The SFD vector of complex $\mathcal{K}$ characterizes the relative frequencies of various simplets in $\mathcal{K}$. By definition, $|\mathcal{S}_{\mathcal{K}}(i)|$ is the number of simplets of type $i$ in $\mathcal{K}$, where $i \in \{1,\ldots,N_m\}$. The frequency, denoted by $\phi_{\mathcal{K}}(i)$, is obtained by dividing $|\mathcal{S}_{\mathcal{K}}(i)|$ by $\sum_{j=1}^{N_m} |\mathcal{S}_{\mathcal{K}}(j)|$. The vector
$(\phi_{\mathcal{K}}(1), \ldots , \phi_{\mathcal{K}}(N_m))$
is called the SFD vector of the $\mathcal{K}$.
In Figure~\ref{fig:sfd}, we show an SFD vector for two sample SCs.
    
\section{Approximating the SFD Vector}
\label{sec:Approximating-SFD}
In this section, we focus on showing that if we have a method for sampling simplets uniformly from an SC, we can have an $(\epsilon, \delta)$-approximation of the SFD vector.
After that, we study an algorithm for simple uniform sampling that is better than a trivial brute-force sampling method.
Consider a collection of independent samples $X^k = {X_1, \ldots, X_k}$ drawn from a distribution $\phi$ over a domain $D$. Here, $\phi(A)$ signifies the probability of selecting an element from the set $A \subseteq D$. The empirical estimation of $\phi(A)$ based on the samples $X^k$ is:
\[ \hat{\phi}^X(A) = \frac{1}{k} \sum_{j=1}^k 1_{A}(X_j), \]
In this equation, $1_{A}(X_j)$ is an indicator function that equals $1$ when $X_j$ belongs to $A$ and equals $0$ otherwise.
Additionally, let $\cR$ be a family of subsets of $D$.
\subparagraph*{\ed-approximation.}
For any given $\epsilon,\delta \in (0,1)$, we say $X \subseteq D$ is an \emph{$(\epsilon, \delta)$-approximation} of $(\cR, \phi)$, if  with a probability of at least $(1-\delta)$, it satisfies
$ sup_{A \in \cR} |\phi(A) - \hat{\phi}^X(A)| \leq \epsilon$.
\begin{figure*}
    \begin{center}
        \centerline{\includegraphics[width=1\textwidth,keepaspectratio]{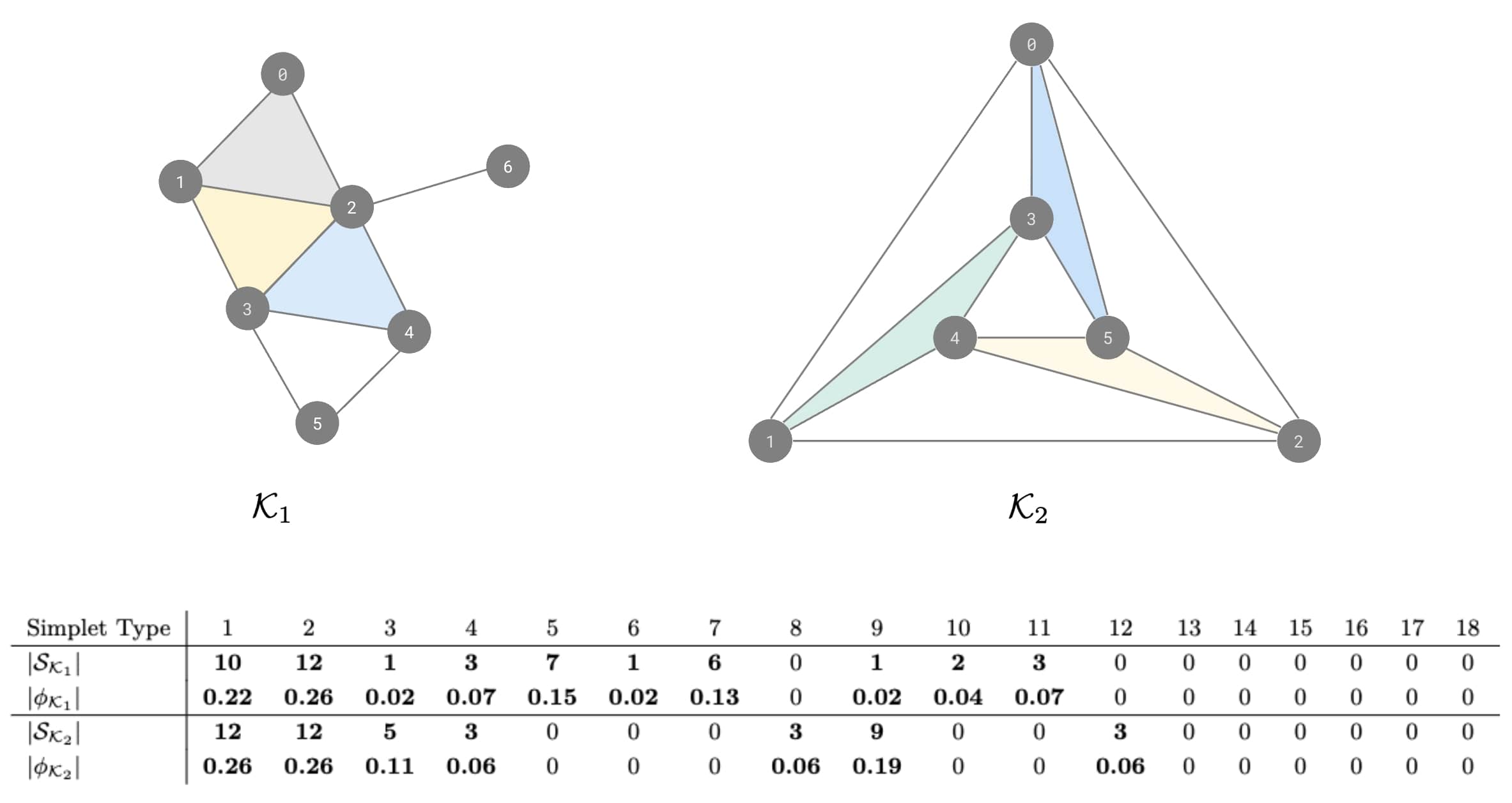}}
	\caption{The SFD vector and the number of at most 4-vertices simplets for two sample SCs. The simplet types in the table refer to the types in Figure~\ref{fig:simplets}.} 
	\label{fig:sfd}
    \end{center}
\end{figure*}
\subsection{Sample Complexity of Approximating the SFD Vector}    
We utilize the concept of Vapnik-Chervonenkis dimension (\VCd), introduced in~\cite{vapnik2015uniform}.
In short, for a domain $D$ and a collection $\cR$ of subsets of $D$, the \VCd $VC(D, \cR)$, represents the maximum size of a set $X \subseteq D$ that can be shattered by $\cR$,
which means $\{ r \cap X | \forall r \in R \} = 2^{|X|}$.
We use \VCd to determine the sample complexity for approximating the SFD vector through simplet sampling models. 
Theorem~\ref{thm:VCSimpletSets} establishes the \VCd of the collection of simplet sets.

\begin{theorem}[VC Dimension of Simplets]
    \label{thm:VCSimpletSets}
    Let $\cR= \{ \mathcal{S}_i~|~1 \leq i \leq N_m \}$ be a family of all simplet sets where $N_m$ is the number of simplet types with at most $m$ vertices, and $D = \mathcal{S}_\mathcal{K}^m$. Then, we have $VC(D, \cR) = 1$.
\end{theorem}

\begin{proof}
We show that a set $X$ with $|X| > 1$ can not be shattered with $(D,\cR)$.
    Let $X$ be a set of simplets shattered with $(D,\cR)$, and assume that $|X| > 1$. Let $s_1$ and $s_2$ be two distinct elements of $X$. There are two possibilities. If elements $s_1$ and $s_2$ belong to the same simplet type, then, set $\{s_1\}$ can not be shattered because there is no set $\mathcal{S}_i$, including only $s_1$. Otherwise, elements $s_1$ and $s_2$ belong to different simplet types, but then $\{s_1,s_2\}$ can not be shattered because no set $\mathcal{S}_i$ contains both.
    Clearly every singleton set can be shattered by one of the $\mathcal{S}_i$s, hence $VC(D, \cR) = 1$.
\end{proof}

The subsequent theorem from \cite{riondato2016fast} illustrates the relationship between the upper bound on the sample complexity of sampling-based $(\epsilon, \delta)$-approximations and \VCd.
\begin{theorem} \label{thm:VCBoundOnApproximation}
Let $D$ be a domain and $\cR$ be a family of subsets of $D$,  with $VC(D,\cR) \leq d$ and $\phi$ be a distribution on $D$. For every $\epsilon, \delta \in (0,1)$, every set $X$ of independent samples drawn from $D$ using $\phi$ that satisfies
\[ |X| \geq \frac{c}{\epsilon^2} \left( d + \ln{\frac{1}{\delta}} \right), \] is an \ed-approximation of $(\cR, \phi)$ for some positive constant $c$. 
\end{theorem}

Combining Theorem~\ref{thm:VCSimpletSets} and Theorem~\ref{thm:VCBoundOnApproximation} we conclude our main result.

\begin{proposition}
\label{prop:SFD-sample-complexity}
Let $X$  be a set of at least $ \frac{c}{\epsilon^2} (1 + \ln{\frac{1}{\delta})}$ simplets sampled uniformly from simplicial complex $\mathcal{K}$.
Then, $X$ obtains an $(\epsilon,\delta)$-approximation on the SFD vector of $\mathcal{K}$.
\end{proposition}




Proposition~\ref{prop:SFD-sample-complexity} shows that we can approximate the SFD vector using sampling-based algorithms, and the sample complexity of these approximations are independent of the simplicial complex size.
This property suggests the usage of approximation algorithms for various simplicial complex sizes with the same sample complexity.

\subsection{Simplet Uniform Sampling Algorithm}
\label{sec:Sampling-alg}
In this section, we propose a uniform sampling algorithm for simplets in a connected simplicial complex $\mathcal{K}$ that is better than a trivial brute-force method. The algorithm we present is a Monte-Carlo Markov-Chain algorithm~\cite{watts1998collective}, that samples sufficiently many simplets uniformly at random. We assume that $\mathcal{K}$ is connected with at least three vertices, and $m\ge3$.

For the sampling part, we perform a random walk on a directed graph $\mathcal{P}_\mathcal{K}^m$ whose vertex set (states) is a set of all simplets in complex $\mathcal{K}$ with at most $m$ vertices. Out-neighbors of every state $s$ can be created by adding one vertex to $s$, removing one vertex from $s$, or replacing one vertex in $s$ with another vertex out of $s$.

The transition probability matrix $T$ for the random walk is such that every cell $T(i,j)$ defines the transition probability from state $i$ to $j$. If $i$ and $j$ are not neighbors, we set $T(i,j) = 0$. Otherwise, we set $T(i,j) = min(\frac{1}{d(i)},\frac{1}{d(j)})$ where $d(i)$ specifies the number of out-neighbors of state $i$. Also, for every $i$, if the sum of transitions from $i$ is not equal to $1$, we allocate the remaining probability to a self-loop for $i$. Observe that since $\mathcal{K}$ is finite, $\mathcal{P}_\mathcal{K}^m$ is finite and since $\mathcal{K}$ is connected, the random walk is irreducible. Indeed, since $\mathcal{K}$ is connected, there is a vertex $u$ in $\mathcal{K}$ that is connected to at least two other vertices $v,w$. So the three simplets on $\{u,v,w\}$, on $\{u,v\}$, and on $\{u,w\}$ form a triangle in $\mathcal{P}_\mathcal{K}^m$ with positive probabilities on the edges, this means the random walk is aperiodic. Also $T$ is symmetric, meaning $T = T^T$. This ensures that the random walk on $\mathcal{P}_\mathcal{K}^m$ converges to the uniform stationary distribution $(\frac{1}{|\mathcal{S}_\mathcal{K}^m|},\ldots,\frac{1}{|\mathcal{S}_\mathcal{K}^m|})$. So, using this random walk on $\mathcal{P}_\mathcal{K}^m$, we can select a simplet from the input complex $\mathcal{K}$ with uniform distribution.

\subsection{The SFD Vector Approximation Algorithm}
\label{sec:SFD-approx-alg}
Now, we propose the $(\epsilon,\delta)$-approximation algorithm on the SFD vector of $\mathcal{K}$. For input $\epsilon, \delta \in (0,1)$ and simplicial complex $\mathcal{K}$, first the algorithm calculates the number $\ell$ of samples needed, according to Proposition~\ref{prop:SFD-sample-complexity}. After that it executes $\ell$ times the sampling algorithm, presented above, to find the set $X$ of $\ell$ simplets that are chosen uniformly at random. Based on $X$, it computes $\hat{\phi}^X_{\mathcal{K}}(i)$, that is a $(\epsilon,\delta)$-approximation for $\phi_{\mathcal{K}}(i)$, for $1 \leq i \leq N_m$.  The vector $(\hat{\phi}^X_{\mathcal{K}}(1), \ldots , \hat{\phi}^X_{\mathcal{K}}(N_m))$ is therefore a $(\epsilon,\delta)$-approximation for the SFD vector of $\mathcal{K}$.


\subparagraph*{Time Complexity of the SFD vector Approximation Algorithm}
The time complexity of the $(\epsilon,\delta)$-approximation algorithm, consists of two components: the number of samples and the time complexity for sample identification. Having established that $O(\frac{1}{\epsilon^2} \;\cdot\; (1 + \ln{\frac{1}{\delta})})$ samples are sufficient for $(\epsilon,\delta)$-approximation, our focus shifts to analyzing the time complexity of the MCMC sampling algorithm. The mixing time $t_{mix}^G$ in a random walk on graph $G$ is the number of steps needed to be close to its stationary state with high probability. Lemma~\ref{lemma:deltainmarkovchain} limits the maximum degree of $\mathcal{P}_{\mathcal{K}}^m$, and then Lemma~\ref{lemma:mixingtime} shows an upper bound on $t_{mix}^{\mathcal{P}_{\mathcal{K}}^m}$ in terms of the number of vertices $n$, the maximum degree $\Delta$ in $\mathcal{K}$, and the diameter $diam(\mathcal{K})$, which is the length of maximum shortest path between any pair of vertices in $\mathcal{K}$.

\begin{lemma}
\label{lemma:deltainmarkovchain}
The maximum degree of $\mathcal{P}_{\mathcal{K}}^m$ satisfies $\Delta(\mathcal{P}_{\mathcal{K}}^m) \in O(m^2 \cdot \Delta)$.
\end{lemma}
\begin{proof}
We can create neighbors of every state in $\mathcal{P}_\mathcal{K}^m$ by adding a new vertex, removing a vertex, or replacing two vertices. The number of neighbors by adding a new vertex is at most $m \cdot \Delta$, by removing a vertex is at most $m$, and by replacing two vertices is at most $m \cdot (m-1) \cdot \Delta$.
Therefore, the maximum degree of every state in $\mathcal{P}_\mathcal{K}^m$ is in $O(m^2 \cdot \Delta)$.
\end{proof}

\begin{lemma}
\label{lemma:mixingtime}
The mixing time of the markov chain on $\mathcal{P}_{\mathcal{K}}^m$ is in $O(\log(n) \cdot \Delta \cdot diam(\mathcal{K})^2)$.

\end{lemma}
\begin{proof}
Theorems (12.4) and (13.26) in~\cite{levin2017markov} imply that the mixing time $t^G_{mix}$ of a random walk on graph $G$ with $n$ vertices is 
in $O(\log(n) \cdot \Delta(G) \cdot diam(G)^2)$.
For $\mathcal{P}_{\mathcal{K}}^m$ we can reach from every state $i$ to every other state $j$ with $diam(\mathcal{K})+m$ steps as follows: Assume $v \in V(i), u \in V(j)$ and assume a shortest path from $v$ to $u$ in $\mathcal{K}$. Starting from $i$, in every step, we replace one vertex from the current state with the unused closest vertex to $v$ in the shortest path from $v$ to $u$, until we reach $u$. After that we replace vertices that are not in $j$ with vertices in $j$, starting from neighbors of $u$. We make sure that after each step the simplet remains connected. So, $diam(\mathcal{P}_{\mathcal{K}}^m) = diam(\mathcal{K})+m$ and therefore in the markov chain $\mathcal{P}_{\mathcal{K}}^m$ we have

$$t^{\mathcal{P}_{\mathcal{K}}^m}_{mix} \in O(log(n(\mathcal{P}_{\mathcal{K}}^m)) \cdot \Delta(\mathcal{P}_{\mathcal{K}}^m) \cdot diam(\mathcal{P}_{\mathcal{K}}^m)^2) \nonumber \in O(log(n) \cdot \Delta \cdot diam(\mathcal{K})^2).$$
\end{proof}

\begin{corollary}[Time Complexity of $(\epsilon,\delta)$-approximation of SFD vector]
\label{corollary:timecomplexity}
Let $\mathcal{K}$ be a simplicial complex with the number of vertices $n$, maximum degree $\Delta$ and diameter $diam(\mathcal{K})$. The time complexity of $(\epsilon,\delta)$-approximation of SFD vector of $\mathcal{K}$ is $O(\frac{1}{\epsilon^2} \cdot (1 + \ln{\frac{1}{\delta})} \cdot \log(n) \cdot \Delta \cdot diam(\mathcal{K})^2).$
\end{corollary}
In practice, for a large sparse simplicial complex $\mathcal{K}$, since $\Delta$ and $diam(\mathcal{K})$ are bounded, the above bound is sublinear in the size of $\mathcal{K}$ (i.e. the number of vertices or $\mathcal{K}$).

\subparagraph*{Implementation and Experiments.}
We implement an algorithm for counting the exact number of simplets of different types and another algorithm for approximating the frequencies based on uniform simplet sampling,
with their source code accessible on GitHub~\cite{network-signature}.
This experimental outcome demonstrates that the confidence in the \ed-approximation is unrelated to the size of the input complex.

\section{Conclusion}
\label{sec:Conclusions} 
This paper introduced the Simplet Frequency Distribution (SFD) vector and a method for approximating it with simplet sampling algorithms. Also, we studied the sample complexity of approximating the SFD vector for SCs and showed that the obtained bounds are independent of SC size.
We also showed that we can approximate the SFD vector with a specific error and confidence, and the time complexity depends only on the time complexity of sampling algorithm for finding a sample, that is sublinear in the algorithm we presented. It would be beneficial to have such algorithms with time complexity that is independent of the SC size.

Combining these approaches with filtrations of simplicial complexes and exploring them within alpha complexes would be interesting. 
Additionally, defining the vertex-specific SFD vectors for each vertex in the complex could offer valuable insights into their potential to convey more information about the global structure of the complex.


\bibliography{eurocg24}






\end{document}